\documentclass[onecolumn,journal]{IEEEtran}

\usepackage{amsthm}
\usepackage{thmtools}
\usepackage{thm-restate}

\usepackage{amsmath,amsfonts}
\usepackage{algorithmic}
\usepackage{algorithm}
\usepackage{array}
\usepackage[caption=false,font=normalsize,labelfont=sf,textfont=sf]{subfig}
\usepackage{textcomp}
\usepackage{stfloats}
\usepackage{url}
\usepackage{verbatim}
\usepackage{graphicx}
\usepackage{cite}

\makeatletter
\let\IEEEproof\proof 
\let\IEEEendproof\endproof
\let\proof\@undefined
\let\endproof\@undefined
\newcommand*{\etc}{
    \@ifnextchar{.}
        {etc}
        {etc.\@\xspace}
}
\makeatother

\usepackage{report}

\let\proof\IEEEproof
\let\endproof\IEEEendproof

\def\HG{H_{\text{G}}}
\def\HC{H_{\text{C}}}
\def\PCG{P_{\text{C}|\text{G}}}
\def\PGC{P_{\text{G}|\text{C}}}
 \newcommand{\HypGC}{\begin{array}{c} {\scriptstyle \HG} \\ \geq \\ < \\ {\scriptstyle \HC} \end{array}} 

\newcommand{\eps}{\varepsilon}

\def\aw{\mathrm{a}_\mathcal{W}}
\usepackage{cite, float}
\usepackage{mathtools}

\begin{document}
\title{Testing the Isotropic Cauchy Hypothesis} 

\author{%
 \IEEEauthorblockN{Jihad Fahs,~\IEEEmembership{Member,~IEEE,} Ibrahim Abou-Faycal~\IEEEmembership{Member,~IEEE,} and Ibrahim Issa~\IEEEmembership{Member,~IEEE}}\\
 \IEEEauthorblockA{Department of Electrical and Computer Engineering \\
                   American University of Beirut, Beirut, Lebanon \\
                   Emails: \{Jihad.Fahs, Ibrahim.Abou-Faycal, Ibrahim.Issa\}@aub.edu.lb}
}

    

\maketitle


\begin{abstract}
Isotropic $\alpha$-stable distributions are central in the theory of heavy-tailed distributions and play a role similar to that of the Gaussian density among finite second-moment laws. Given a sequence of $n$ observations, we are interested in characterizing the performance of Likelihood Ratio Tests where two hypotheses are plausible for the observed quantities: either isotropic Cauchy or isotropic Gaussian. Under various setups, we show that the probability of error of such detectors is not always exponentially decaying with $n$ with the leading term in the exponent shown to be logarithmic instead and we determine the constants in that leading term. Perhaps surprisingly, the optimal Bayesian probabilities of error are found to exhibit different asymptotic behaviors. 


\end{abstract}

\section{Introduction}
\IEEEPARstart{D}{etecting} the underlying sampling distribution of an empirical set of $n$-observations is a challenging and fundamental question in decision theory. It falls under the family of hypothesis testing problems whenever deciding on a possibility among a given number of choices. In the binary scenario ($p$ vs. $q$), it is well known that  Likelihood Ratio Tests (LRT) detectors are optimal for multiple setups. Under a Neyman-Pearson criterion or under a Bayesian risk formulation, the  performance of such detectors is expected to be exponentially decaying with $n$, with the exponent being a Kullback-Leibler (KL) distance between pertinent laws. This is the case in a discrete scenario~\cite[Th. 11.8.3 \& Th. 11.9.1]{cover} where a well-known proof is based on the method of types. Extensions to continuous setups do exist under certain conditions~\cite{tusnady1977asymptotically,csiszar1998}. 

One interesting quest is to analyze this decision problem while testing for 
\begin{enumerate}   
\item the Cauchy distribution: representative of the family of ``heavy-tailed" $\alpha$-stable distributions versus 
\item the Gaussian distribution: a central model among finite-variance probability laws. 
\end{enumerate}

Note that both models are limiting distributions of Independent and Identically Distributed (IID) sums of Random Variables (RV)s by the results of the generalized Central-Limit Theorem~\cite{kolmo} and therefore supported by strong theoretical evidence. Additionally, isotropic\footnote{We will use ``circular'' interchangeably with ``isotropic'' in this document.} $\alpha$-stable distributions are found to be accurate models for empirical data in a multitude of scenarios in which they are shown to provide a better fit to the data under consideration~\cite{nbrm2001,kuzer2004,nolancirc2013}. 

In this work, we tackle the problem of testing the $\alpha$-stable hypothesis by considering the $n$-dimensional circularly-symmetric Cauchy distribution. More specifically, given a sample of $n$ observations sampled from a rotationally invariant density, we want to characterize the error limits of the LRT devices detecting whether the observations are circular Gaussian --with Probability Density Function (PDF) denoted $p_G(\cdot)$, or circular Cauchy --with PDF $p_C(\cdot)$. We note that the proposed setup is substantially different from the standard one regarding two main aspects:
\begin{enumerate}
\item First, circular Cauchy vectors are correlated and not independent as is the case for a Gaussian vector. 
\item Second, the relative entropy $D(p_C\|p_G)$, 
is infinite as shown in~Appendix~\ref{app:KL}. 
\end{enumerate}
Correspondingly, the results are found to deviate from the norm regarding two aspects:
\begin{itemize}
    \item[$\bullet$] First, we show that the probability of error is not exponentially small with $n$. Namely we prove that, under a Neyman-Pearson formulation and in the first order of the exponent, the leading term is linear for the error of the first kind ($\PCG$) but only logarithmic for that of the second kind ($\PGC$). A similar logarithmic behavior of the exponent of the leading term is shown for the Bayesian error. Whenever possible, we determine the multiplicative constants of the provided asymptotic expressions. 
    \item[$\bullet$] Second, we observe that the two types of Bayesian errors do not possess the same asymptotic rate which is at odds with the general assumption of a similar asymptotic behavior under IID statistics (see~\cite[Th. 11.9.1]{cover} for the discrete case). We postulate that this is due to the fact that the observations are correlated and we support our observation by providing in Appendix~\ref{app:corr} illustrative examples in which one of the two hypotheses is correlated.
\end{itemize}

The rest of this paper is organized as follows: we present the problem in Section~\ref{sec:PF} and the main results in the form of two theorems in
Section~\ref{sec:main} along with appropriate numerical evaluations. Derivation of the results use
technical lemmas and are detailed in Section~\ref{sec:pre}. Finally,
 Section~\ref{sec:con} concludes the paper. 

\paragraph*{Notations}
We adopt the following conventions. For two asymptotically positive functions, $f(n) = o \left(g(n)\right)$ if for every $c > 0$, $f(n) \leq c g(n)$ for $n$ large enough. We write $f(n) = \Theta\left(g(n)\right)$ if there exist $c_1 > 0$ and $c_2 > 0$ such that $f(n) \leq c_1 g(n)$ and $f(n) \geq c_2 g(n)$ for $n$ large enough. 
Functions
$f(n)$ and $g(n)$ are said to be ``asymptotically equivalent'' if $ \lim_{n \to \infty} \frac{f(n)}{g(n)} = 1$; we write $f(n) \sim g(n)$. 

\section{Problem Formulation}
\label{sec:PF}
 We consider the binary hypothesis testing problem where, given a (large) number $n > 0$ of observations, we need to decide between the two hypotheses:
\begin{eqnarray*}
    && \left\{ \begin{array}{ll}
        \displaystyle  H_0 \equiv \HG: &{\bf Y}^n \,\, \text{circular} \,\,  \mathcal{N}(0,\sigma^2), \\ 
        \displaystyle  H_1 \equiv \HC: &{\bf Y}^n \,\, \text{circular} \,\,  \mathcal{C}(0,\gamma).
      \end{array} \right.
  \end{eqnarray*}

Whether under a Bayesian risk formulation, or a  Neyman-Pearson criterion, the optimal device is known to be a LRT: 
\begin{equation}
\frac{p_{{\bf Y}|H} \left( {\bf y}|\HG \right)}{p_{{\bf Y}|H}\left({\bf y}|\HC \right)} = \frac{p_G({\bf y})}{p_C({\bf y})} \HypGC \eta,
\label{NPdet}
\end{equation}
for some appropriately chosen $\eta >0$. Note that we identify the ``Gaussian hypothesis" $\HG$ with the null hypothesis $H_0$ whenever we use the "false alarm", "detection" and "error types" terminology in this work. As such, the probability of false alarm for example $P_F$ is equal to $P_F = \PCG$ and as $\eta$ increases, $P_F$ and $P_D$ increase.

Our main objective is twofold. First we study the optimal Bayesian device and characterize its overall probability of error at large values of $n$ via computing the asymptotic behavior of the probabilities of detection $P_D$ and false alarm $P_F$. More specifically, if $\pi_G$ and $\pi_C = 1 - \pi_G$ denote the prior probabilities of $H_0$ and $H_1$ respectively, Theorem~\ref{th:cherinfo} shows that as $n \to \infty$, the probability of error of the optimal device is
\begin{equation*}
P_{\text{e}} = \pi_G P_F + \pi_C (1 - P_D) = \Theta\left(\sqrt{\frac{\ln n}{n}}\right).
\end{equation*}

Second, we adopt the Neyman-Pearson criterion and answer the following question: let $P_F = \PCG = \eps$, for some $\eps \in (0, 1)$. What is the asymptotic behavior of $1 - P_D = \PGC$ of the optimal device as $n \rightarrow \infty$? We show in Theorem~\ref{th:error} that
\begin{equation*}
1 - P_D = \PGC = \Theta\left(\frac{1}{\sqrt{n}}\right).
\end{equation*}

We also answer the same question for the inverted problem: if $1 - P_D = \eps$, we show in Theorem~\ref{th:error} that 
\begin{equation*}
 P_F = \PCG = \frac{e^{-\Theta(n)}}{\sqrt{n}}.
\end{equation*}



For mathematical convenience and without loss of generality, we scale our observation vector ${\bf y}^n$ by $\frac{1}{\gamma}$ and we denote {\em the constant\/} $\sqrt{2} \frac{\sigma}{\gamma}$ by $\xi = \sqrt{2} \frac{\sigma}{\gamma}$. The resulting equivalent problem is: 
 \begin{eqnarray*}
    && \left\{ \begin{array}{ll}
        \displaystyle  H_0 \equiv \HG: &{\bf Y}^n \,\, \text{circular} \,\, \mathcal{N} \left(0, \frac{\xi^2}{2} \right)\\ 
        \displaystyle  H_1 \equiv \HC: &{\bf Y}^n \,\, \text{circular} \,\, \mathcal{C}(0,1).
      \end{array} \right.
  \end{eqnarray*}

The PDFs $p_{{\bf Y}|H}({\bf y}|H_0)$ and $p_{{\bf Y}|H}({\bf y}|H_1)$ are given by
\begin{align}
p_{{\bf Y}|H}({\bf y}|H_0) &= p_{{\bf Y}|H}({\bf y}|\HG) = p_G({\bf y}) = \frac{1}{\pi^{\frac{n}{2}} \xi^n} e^{-\frac{r^2}{\xi^2}}, \nonumber\\[0.2em]
p_{{\bf Y}|H}({\bf y}|H_1) &= p_{{\bf Y}|H}({\bf y}|\HC) \! = \! p_C({\bf y})  \! = \! \frac{\Gamma\!\left(\frac{n + 1}{2}\right)}{\pi^{\frac{n + 1}{2}} } \frac{1}{\left(1 + r^2\right)^{\frac{n + 1}{2}}} \label{eq:cauchypdfn}
\end{align}
where $r^2 = \|{\bf y}\|^2$. 

\section{Main Results}
\label{sec:main}

We state our main results in the form of Theorems~\ref{th:cherinfo} and~\ref{th:error} and we validate them numerically in Section~\ref{sec:Num}. 

In a Bayesian setting, the probably of error is determined in the following theorem:
\begin{restatable}[]{theorem}{Bayesian}
\label{th:cherinfo}
Consider a sequence of $n$ observations and a binary decision problem with two hypotheses for their distribution: either circular (IID) Gaussian $\mathcal{N}(0,\sigma^2)$ with prior probability $0 < \pi_G < 1$, or circular (correlated) Cauchy $\mathcal{C}(0,\gamma)$ with prior probability $ \pi_C = 1 - \pi_G$. 

Let $\xi \eqdef \sqrt{2}\frac{\sigma}{\gamma}$, $\tilde{\eta} \eqdef \frac{\pi_C}{\pi_G}$ and $C \eqdef \sqrt{\frac{1}{\tilde{\eta}} \, \frac{\xi}{2}} \, e^{\frac{1}{2 \xi^2}}$. For the Maximum A Posteriori (MAP) detector --optimal under a Bayesian formulation, 
\begin{align}
1- P_D & \sim \sqrt{\frac{2}{\pi}} \, \frac{1}{\tilde{\eta}}\sqrt{\frac{\ln \big( C n \big)}{ \big( C n \big)}} \nonumber \\
P_F & \sim \sqrt{\frac{2}{\pi}} \, \frac{1}{\sqrt{ \big( C n \big) \ln \big( C n \big)}} \nonumber \\
P_e = \pi_G P_F + \pi_C (1 - P_D) & \sim \sqrt{\frac{2}{\pi}} \, \frac{\pi_C}{\tilde{\eta}}\sqrt{\frac{\ln \big( C n \big)}{ \big( C n \big)}}. \label{eq:cherinfo}
\end{align}
\end{restatable}

Note that the optimal error probabilities $P_F$ and $(1 - P_D)$ --the ones that minimize the Bayesian probability of error $P_{\text{e}}$, exhibit an asymptotic behavior that is {\em not symmetric\/}. That symmetry is a known property of the optimal Bayesian detector under IID observations. This is not true in general under correlated settings as in Theorem~\ref{th:cherinfo}. We provide additional examples in Appendix~\ref{app:corr} which also show that there are no guarantees under the correlated setup for $P_{\text{e}}$ to become infinitesimally small as the number of observations $n$ increases. 

\begin{restatable}[]{theorem}{NP}
\label{th:error}
Consider a sequence of $n$ observations and a binary decision problem with two hypotheses for their distribution: either circular (IID) Gaussian $\mathcal{N}(0,\sigma^2)$ or circular (correlated) Cauchy $\mathcal{C}(0,\gamma)$. Let $\xi \eqdef \sqrt{2}\frac{\sigma}{\gamma}$ and consider an $0 < \eps < 1$. 

Under a Neyman-Pearson criterion, the optimal LRT detector is such that
\begin{itemize}
\item Case 1: whenever $P_F = \eps$, $(1 - P_D) \sim \kappa_0(\epsilon) \frac{1}{\sqrt{n}}$, where $\kappa_0(\epsilon) \eqdef \sqrt{\frac{2}{\pi}} \frac{2}{\xi}e^{-\frac{1}{\xi^2}}Q^{-1}\left(\frac{\epsilon}{2}\right)$.
\item Case 2: whenever $(1- P_D) = \eps$, 
\begin{equation*}
\frac{\kappa_1(\epsilon)e^{- \kappa_2(\epsilon) n }}{\sqrt{n}} \leq P_F \leq \frac{\kappa_3(\epsilon)e^{- \kappa_4(\epsilon) n }}{\sqrt{n}},
\end{equation*}
for some $\kappa_1(\epsilon), \kappa_2(\epsilon), \kappa_3(\epsilon), \kappa_4(\epsilon) > 0$.
\end{itemize}
\end{restatable}

The derivations of the theorems are presented in Section~\ref{sec:pre}.

\subsection{Numerical Simulations}
\label{sec:Num}

We compute the performance of the LRT defined in~(\ref{NPdet}) by evaluating the different quantities investigated in Theorems~\ref{th:cherinfo} and~\ref{th:error}. For the results presented hereafter we use  $\xi = \sqrt{2}$ and dimension $n \in [11,340]$. Fig.~\ref{figure1} depicts the behavior of the Bayesian probability of error $P_e$ function of $n$ for different values of the aprioris. Fig.~\ref{figure1} shows a rather fast convergence of the approximate expression presented in Theorem~\ref{th:cherinfo} to the numerically evaluated $P_e$~\eqref{eq:cherinfo}. In Fig.~\ref{figure2}, we consider Case 1 in Theorem~\ref{th:error} and plot the numerically evaluated  error of the second kind $(1 - P_D)$ versus $n$ in addition to the approximate expression for different values of $\eps$. Again, the fast convergence of the approximation is observed. Finally, we depict in Fig.~\ref{figure3} $-\ln P_F$ as a function of $n$. A linear behavior is noticed as predicted by the result of Case 2 in Theorem~\ref{th:error}. 

\begin{figure}[htp]
  \begin{center}
    \includegraphics[width=\linewidth]{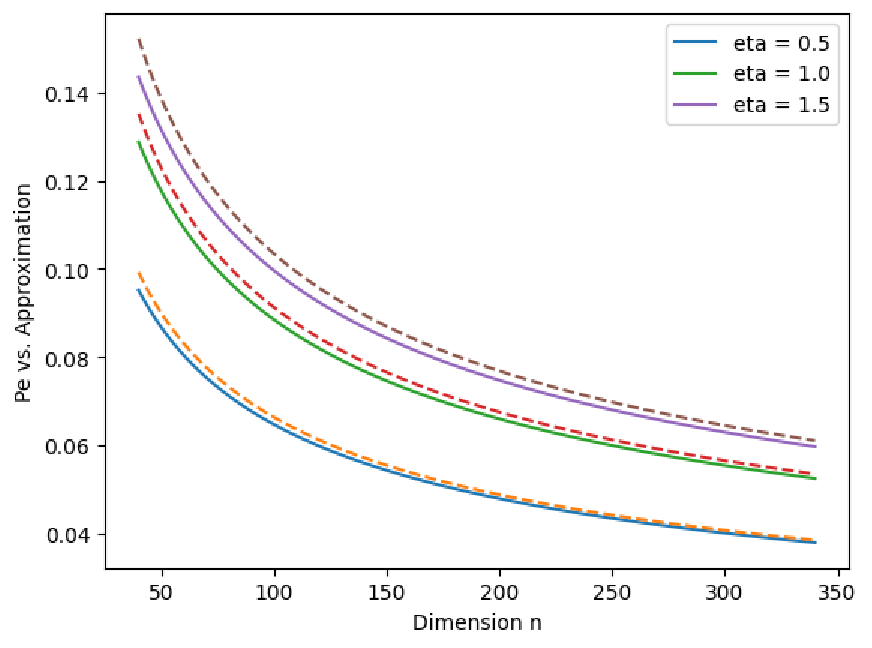}
    \caption{\small $P_e$ vs. $n$ for $\eta = 0.5, 1, 1.5$. Solid line: numerical computations. Dashed line: approximate expression.
      \label{figure1}}
  \end{center}
\end{figure} 

\begin{figure}[htp]
  \begin{center}
    \includegraphics[width=\linewidth]{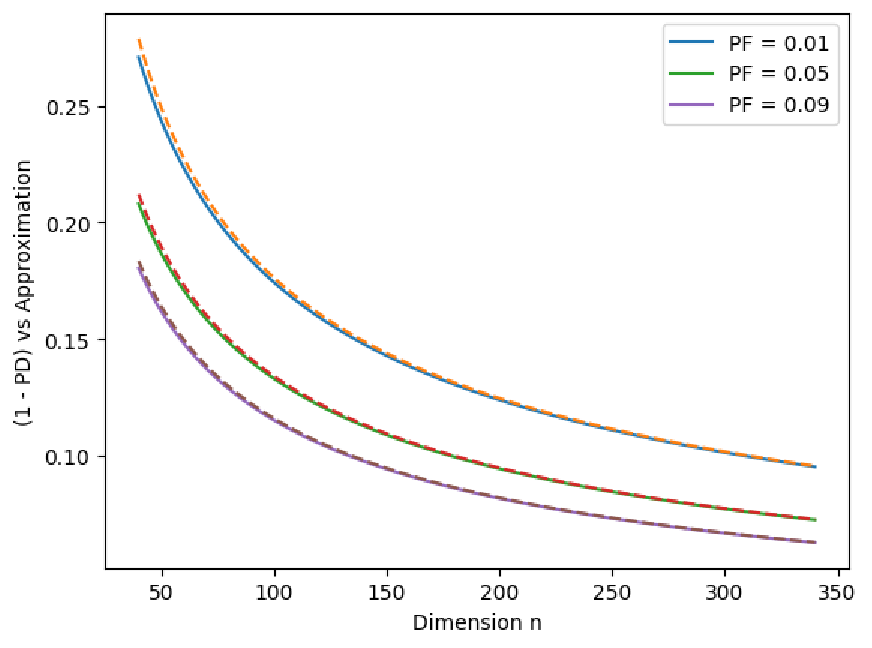}
    \caption{\small $(\PGC)$ vs. n for $P_F = 0.01, 0.05, 0.09$. Solid line: numerical computations. Dashed line: approximate expression.
      \label{figure2}}
  \end{center}
\end{figure} 

\begin{figure}[htp]
  \begin{center}
    \includegraphics[width=\linewidth]{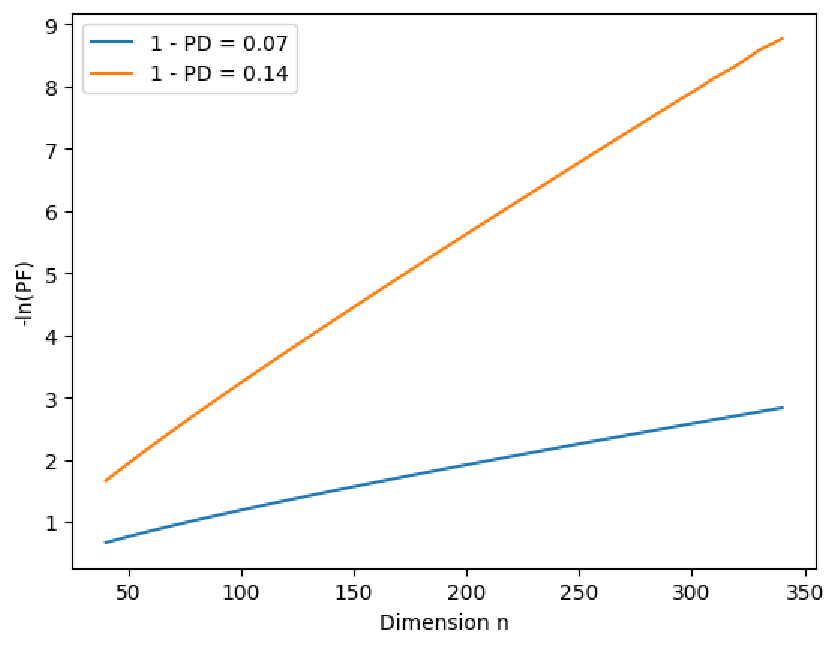}
    \caption{\small $- \ln (\PCG)$ vs. n for $1- P_D = 0.07, 0.14$.
      \label{figure3}}
  \end{center}
\end{figure} 

The remainder of this document is dedicated to deriving the main results stated in Theorems~\ref{th:cherinfo} and~\ref{th:error}.

\section{Technical Proofs}
\label{sec:pre}

\subsection{Preliminaries}

Define the left-hand side of~\eqref{NPdet} as $\ell(r)$:
\begin{equation*}
\ell(r) = \frac{p_{{\bf Y}|H} \left( {\bf y}|\HG \right)}{p_{{\bf Y}|H}\left({\bf y}|\HC \right)} = \frac{\sqrt{\pi}}{\Gamma\!\left(\frac{n + 1}{2}\right) \xi^{n}} e^{-\frac{r^2}{\xi^2}} \left(1 + r^2\right)^{\frac{n + 1}{2}}.
\end{equation*}

Whenever $\ell(r)$ is greater or equal to $\eta$, the LRT~\eqref{NPdet} decides on the Gaussian hypothesis. As such, when computing $P_F$ and $P_D$, the larger the value of $\eta$ is, the larger these probabilities.

We note that when $n > (2/\xi^2-1)$,  $\ell(r)$ satisfies the following properties illustrated in Fig.~\ref{fig:ratio-plot}:
\begin{enumerate}
\item Asymptotically, $\ell(r) \rightarrow 0^{+}$ as $r \rightarrow \infty$.
\item Its $y$-intercept, which we denote by $\eta_0$, is equal to 
\begin{equation}
\label{eq:eta0}
\eta_{0} = \frac{\sqrt{\pi}}{\Gamma\!\left(\frac{n + 1}{2}\right) \, \xi^n}.
\end{equation} 
\item It peaks at 
\begin{equation*}
r_{\max} = \sqrt{\xi^2 \frac{(n+1)}{2} -1},
\end{equation*} 
reaching a value of 
\begin{equation*}
\eta_{\max} = e^{\frac{1}{\xi^2}} \left(\frac{n+1}{2} \, \xi^2 e^{-1} \right)^\frac{n+1}{2} \eta_0.
\end{equation*}
\end{enumerate}

\begin{figure}[htb]
    \centering
    \includegraphics[width=\linewidth]{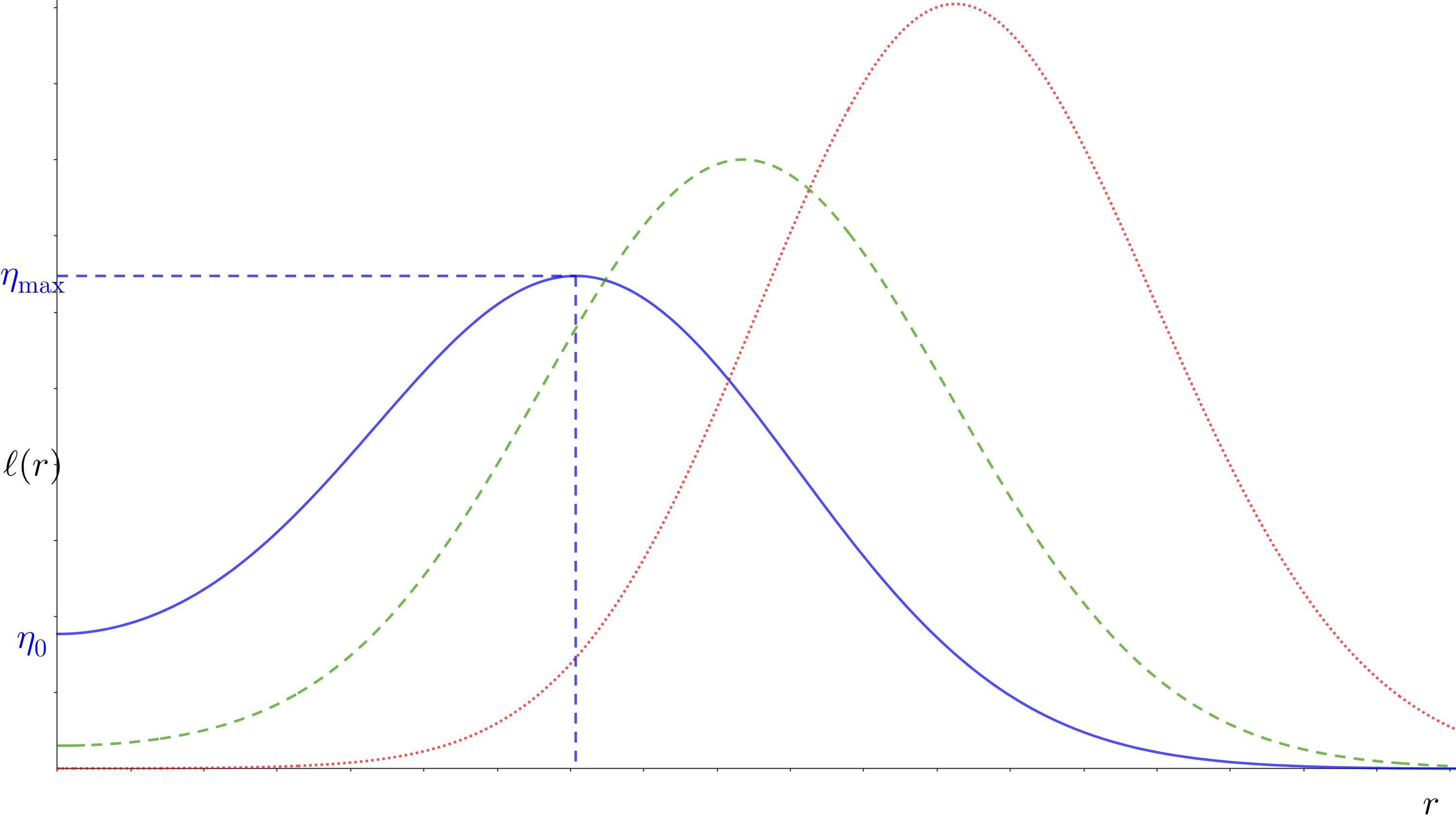}
    \caption{$\ell(r)$ for $\xi^2=1$, with $n=5$ (solid blue), $n=8$ (dashed green), and $n=13$ (pointed red).}
    \label{fig:ratio-plot}
\end{figure}

We assume that this condition $n > (2/\xi^2-1)$ is always satisfied since we are interested in the behavior of detectors under a large number $n$ of observations. 

Using Stirling's approximation, one can show:
\begin{equation}
\eta_0 \sim \frac{(2 e)^{\frac{n-1}{2}}}{\xi^n (n - 1)^{\frac{n}{2}}} \qquad \& \qquad
\eta_{\max} \sim e^{\frac{1}{\xi^2}} \frac{\xi}{2} \, \sqrt{n}, \label{eq:fmaxtail}
\end{equation}
and since $\xi$ is constant, 
\begin{equation}
\eta_0^{\frac{2}{n+1}} \sim \frac{(2 e)}{\xi^{2}} \, \frac{1}{n} \qquad \& \qquad
\eta_{\max}^{\frac{2}{n+1}} \sim 1. \label{eq:limfmax}
\end{equation}

\subsection{Operating Regime (OR)}
\label{sec:or}

Denote by
\begin{equation}
\label{eq:arglw}
\aw(n) \eqdef -e^{-1} \!\left(\frac{\eta}{\eta_{\max}}\right)^{\frac{2}{n+1}} \sim -e^{-1} \eta^{\frac{2}{n+1}},
\end{equation}
where the approximation is due to~\eqref{eq:limfmax} and let $\mathcal{W}_{0}(\cdot)$ and $\mathcal{W}_{-1}(\cdot)$ denote the Lambert W function of branches $0$ and $-1$ respectively. As can be seen in Fig.~\ref{fig:ratio-plot}, there are three regimes based on the value of the threshold $\eta$ in~\eqref{NPdet}: $0 < \eta \leq \eta_0$, $\eta_0 < \eta < \eta_{\max}$, and $\eta \geq \eta_{\max}$. As it will turn out to be the case, the only ``Operating Regime" (OR) for the problem at hand is the second one: $\eta_0 < \eta < \eta_{\max}$, for which $\ell(r) = \eta$ has two solutions,
\begin{align}
r_1^2  & = -\xi^2\left(\frac{n+1}{2}\right)\mathcal{W}_{0} \left(\aw(n)\right) - 1 \label{lw0}\\
r_2^2   & = -\xi^2\left(\frac{n+1}{2}\right)\mathcal{W}_{-1} \left( \aw(n) \right) - 1.  \label{lw-1}
\end{align} 
We next compute $P_D$ and $P_F$ under the OR. We first compute $(1 - P_D)$:
\begin{align}
1- P_D &= \text{Pr}\left(\hat{H} = H_0| H_1\right) = \int_{r_1}^{r_2} p_{{\bf Y}|H}({\bf y}|H_1)\,d{\bf y} \nonumber\\
&= \frac{2}{\sqrt{\pi}} \frac{\Gamma\!\left(\frac{n+1}{2}\right)}{  \Gamma\!\left(\frac{n}{2}\right)} \int_{r_1}^{r_2} \left(1 + r^2\right)^{-\frac{n + 1}{2}} r^{n-1}\,dr \label{eq:pddfR2}\\
&=  \frac{2}{\sqrt{\pi}} \frac{\Gamma\!\left(\frac{n+1}{2}\right)}{  \Gamma\!\left(\frac{n}{2}\right)} \int_{r_1}^{r_2} \left(\frac{1}{r^2} + 1 \right)^{-\frac{n + 1}{2}} r^{-2}\,dr \nonumber\\
&= \frac{1}{\sqrt{\pi}} \frac{\Gamma\!\left(\frac{n+1}{2}\right)}{  \Gamma\!\left(\frac{n}{2}\right)}\int_{\frac{1}{r_2^2}}^{\frac{1}{r_1^2}} \left(1 + u\right)^{-\frac{n + 1}{2}} u^{-\frac{1}{2}}\,du, \label{eq:1-pd}
\end{align}
where we used the change of variable $u = \frac{1}{r^2}$ to write the last equality. As for $P_F$, we have:
\begin{align}
1 - P_F & = \text{Pr}\left(\hat{H} = H_0| H_0\right) = \int_{r_1}^{r_2} p_{{\bf Y}|H}({\bf y}|H_0)\,d{\bf y} \nonumber\\
& = \frac{2}{\xi^n \Gamma\!\left(\frac{n}{2}\right)} \int_{r_1}^{r_2} r^{n-1} e^{-\frac{r^2}{\xi^2}}\,dr \nonumber\\
& = \frac{1}{\Gamma\!\left(\frac{n}{2}\right)}  \int_{\frac{r_1^2}{\xi^2}}^{\frac{r_2^2}{\xi^2}}  r^{\frac{n}{2} - 1} e^{-r}\,dr \nonumber\\
&= \frac{\left(\frac{n}{2} - 1\right)^{\frac{n}{2}}}{\Gamma\!\left(\frac{n}{2}\right)} \int_{\frac{2 r_1^2}{\xi^2(n-2)}}^{\frac{2r_2^2}{\xi^2(n-2)}} e^{\left(\frac{n}{2} - 1\right)(\ln x  - x)} \,dx, \label{eq:changevarLap}
\end{align}
where equation~\eqref{eq:changevarLap} is due to the change of variable $r = (\frac{n}{2} - 1) x$. 
In order to determine the asymptotic behavior of both $P_D$ and $P_F$ at large values of $n$, , it is imperative to first analyze the growth rate of the intersection points under the OR. We identify three cases: 
\begin{itemize}
\item[(i)]  
$\eta^{\frac{2}{n+1}} \underset{n \rightarrow \infty}{\longrightarrow} 1$ : 
This includes, for example, the case of $\eta$ fixed. Equation~(\ref{eq:arglw}) implies that $\aw(n) \underset{n \rightarrow \infty}{\longrightarrow} -e^{-1}$ and~~\cite[p. 350-351]{corless1996}
\begin{align*}
\mathcal{W}_{0}(z) & \sim -1 + \sqrt{2 \, (e \, z + 1)} \\
\mathcal{W}_{-1}(z) & \sim -1 - \sqrt{2 \, (e \, z + 1)},
\end{align*}
which gives 
\begin{align}
\frac{r_1^2}{\xi^2} \sim \frac{n + 1}{2}  \left(1 - \delta_n \right) 
\quad \& \quad
\frac{r_2^2}{\xi^2} \sim  \frac{n+1}{2} \left(1 + \delta_n \right), \label{eq:tailR2r2rmax} 
\end{align}
where
\begin{align}
\delta_n \eqdef \, & \sqrt{2} \, \sqrt{1-\left(\frac{\eta}{\eta_{\max}}\right)^{\frac{2}{n+1}}} \sim 2 \, \sqrt{\frac{\ln \big( \frac{\eta_{\max}}{\eta}\big)}{n}}. \label{eq:deltan}
\end{align}
Note that in this case $\delta_n \underset{n \rightarrow \infty}{\longrightarrow} 0$
since $\lim_{n \rightarrow \infty} \left(\frac{\eta}{\eta_{\max}}\right)^{\frac{2}{n+1}} = 1$.
\item[(ii)]$\eta^{\frac{2}{n+1}} \underset{n \rightarrow \infty}{\longrightarrow} \kappa$, $0 < \kappa <1$: In this case equation~\eqref{eq:arglw} implies that $\aw(n) \rightarrow \left(-\kappa \, e^{-1} \right)$ and:
\begin{eqnarray}
r_1^2  & \sim -\xi^2\left(\frac{n+1}{2}\right)\mathcal{W}_{0} \left(-\kappa \, e^{-1}\right) - 1 \label{eq:caseii0}\\
r_2^2   & \sim -\xi^2\left(\frac{n+1}{2}\right)\mathcal{W}_{-1} \left( -\kappa \, e^{-1} \right) - 1.  \label{eq:caseii-1}
\end{eqnarray}
\item[(iii)] $\eta^{\frac{2}{n+1}} \underset{n \rightarrow \infty}{\longrightarrow} 0$  which implies $\aw(n) \underset{n \rightarrow \infty}{\longrightarrow} 0^-$. Whenever $\aw(n) \rightarrow 0^{-}$, the following approximations hold~\cite[p. 350]{corless1996}:
\begin{eqnarray*}
\mathcal{W}_{0}\left[\aw(n)\right] &\sim& \aw(n)\label{lw0R2}\\
\mathcal{W}_{-1}\left[\aw(n)\right]
&\sim& \ln\left(-\aw(n)\right), \label{lw-1R2} 
\end{eqnarray*}
Thus, equations~(\ref{lw0}) and~(\ref{lw-1}) give 
\begin{eqnarray}
r^2_1 &\sim& \frac{\xi^2 e^{-1}}{2}\,(n + 1)\,\eta^{\frac{2}{n+1}} - 1 \label{tailr1R22}\\
r_2^2 &\sim& \frac{\xi^2 }{2}\, (n + 1) \, \ln \eta^{-\frac{2}{n+1}}. \label{tailr2R22}
\end{eqnarray}
\end{itemize}
For the remaining of this paper, we refer to the preceding three cases of the OR by (R-i), (R-ii), and (R-iii). 


\subsection{The Bayesian Detection Problem; Proof}
We start first by stating a technical lemma. 
\begin{lemma}
\label{lem:asymptotics}
Whenever $\left(\frac{\eta}{\eta_{\max}}\right)^{\frac{2}{n+1}} \underset{n \rightarrow \infty}{\longrightarrow} 1$, the following hold: 
\begin{align*}
1- P_D & \sim \frac{2}{\sqrt{\pi}} \frac{\sqrt{\ln\left(\frac{\eta_{\max}}{\eta}\right)}}{\eta_{\max}}\\
P_F & \sim 2 \, Q\left(\sqrt{2 \ln \left( \frac{\eta_{\max}}{\eta}\right)}\right),
\end{align*}
where $Q(\cdot)$ is the standard Gaussian Q-function. 
\end{lemma}

\begin{proof}
The detector is operating under Regime (R-i). Using the expression of $(1 - P_D)$ in~(\ref{eq:1-pd}), it is shown in equation~\eqref{eq:bdR2fixed} in Appendix~\ref{app:proofeq} that:
\begin{align}
&1- P_D \nonumber\\
& \sim \frac{1}{\sqrt{\pi}} \frac{\Gamma\!\left(\frac{n+1}{2}\right)}{  \Gamma\!\left(\frac{n}{2}\right)} \frac{\sqrt{2}}{\sqrt{n+1}} \left[\Gamma\!\left(\frac{1}{2},\frac{n + 1}{2 r^2_2}\right) - \Gamma\!\left(\frac{1}{2},\frac{n + 1}{2 r^2_1}\right)\right], \nonumber
\end{align}
which gives
\begin{align}
& 1- P_D \nonumber\\
&\sim \sqrt{\frac{2}{\pi}} \frac{\Gamma\!\left(\frac{n+1}{2}\right)}{  \Gamma\!\left(\frac{n}{2}\right)} \frac{1}{\sqrt{n+1}} \left[ \frac{n + 1}{2 r_2^2} - \frac{n + 1}{2 r_1^2}\right] \Gamma^{'}\!\!\left(\frac{1}{2},\frac{1}{\xi^2}\right) \label{eq:taylorGamma}\\
&\sim \sqrt{\frac{2}{\pi}} \frac{\Gamma\!\left(\frac{n+1}{2}\right)}{ \Gamma\!\left(\frac{n}{2}\right)} \frac{1}{\xi^2 \sqrt{n+1}} \, \frac{2 \delta_n }{1 - \delta_n^2} \, \xi e^{-\frac{1}{\xi^2}} \label{eq:Gammaderiv}\\
&\sim \frac{4}{\sqrt{\pi}} \frac{1}{\xi} e^{-\frac{1}{\xi^2}} \sqrt{\frac{\ln\left(\frac{\eta_{\max}}{\eta}\right)}{n}} \label{eq:deltanbehav}\\
&\sim \frac{2}{\sqrt{\pi}} \frac{\sqrt{\ln\left(\frac{\eta_{\max}}{\eta}\right)}}{\eta_{\max}}, \label{eq:finalanswer}
\end{align}
where we used equation~(\ref{eq:fmaxtail}) in order to write equation~(\ref{eq:finalanswer}). Equation~\eqref{eq:taylorGamma} is due to Taylor's approximation of the upper incomplete Gamma function $\Gamma\!\left(\frac{1}{2},x\right)$ around $x = \frac{1}{\xi^2}$ and we used the fact that $\frac{d}{dx} \Gamma(m,x) = -x^{m-1} e^{-x}$ to write equation~\eqref{eq:Gammaderiv}. Equation~\eqref{eq:deltanbehav} is due to equation~\eqref{eq:deltan} and to the fact that
\begin{equation*}
\frac{\Gamma\!\left(\frac{n+1}{2}\right)}{\Gamma\!\left(\frac{n}{2}\right)} \sim \sqrt{\frac{n}{2}}.
\end{equation*}
This concludes the first part of the proof. When it comes to $(1 - P_F)$, its expression is given by~(\ref{eq:changevarLap}): 
\begin{align}
&1 - P_F 
=  \frac{\left(\frac{n}{2} - 1\right)^{\frac{n}{2}}}{\Gamma\!\left(\frac{n}{2}\right)} \int_{\frac{2 r_1^2}{\xi^2(n-2)}}^{\frac{2r_2^2}{\xi^2(n-2)}} e^{\left(\frac{n}{2} - 1\right)(\ln x  - x)} \,dx \nonumber\\
&\sim \frac{\left(\frac{n}{2} - 1\right)^{\frac{n}{2}} e^{-(\frac{n}{2} - 1)}}{\Gamma\!\left(\frac{n}{2}\right)} \int_{\frac{2 r_1^2}{\xi^2(n-2)}}^{\frac{2r_2^2}{\xi^2(n-2)}}   e^{-\left(\frac{n}{2} - 1\right)\frac{(x-1)^2}{2}} \,dx \label{eq:Lapmeth}\\
&=\frac{\left(\frac{n}{2} - 1\right)^{\frac{n}{2}} e^{-(\frac{n}{2} - 1)}}{\Gamma\!\left(\frac{n}{2}\right)} \sqrt{\frac{4\pi}{n - 2}} \nonumber \\
&\qquad \qquad \qquad  \left[ Q\!\left(\frac{\frac{2 r_1^2}{\xi^2(n-2)}- 1}{\sqrt{\frac{2}{n - 2}}}\right) -Q \left(\frac{\frac{2 r_2^2}{\xi^2(n-2)}- 1}{\sqrt{\frac{2}{n - 2}}}\right)\right] \nonumber\\
&\sim 1 - 2 Q\left(\sqrt{2 \ln \left( \frac{\eta_{\max}}{\eta}\right)}\right), \label{eq:tailQ}
\end{align}
where in order to write the last equation we used Stirling's approximation $\Gamma\!\left(\frac{n}{2}\right) \sim \sqrt{2 \pi \left(\frac{n}{2} - 1\right)} \left(\frac{\frac{n}{2} - 1}{e}\right)^{\frac{n}{2}-1}$ and the fact that: 
\begin{eqnarray*}
\frac{\frac{2 r_1^2}{\xi^2(n-2)}- 1}{\sqrt{\frac{2}{n - 2}}} &\sim& -\sqrt{2 \ln \left( \frac{\eta_{\max}}{\eta}\right)} \\
\frac{\frac{2 r_2^2}{\xi^2(n-2)}- 1}{\sqrt{\frac{2}{n - 2}}} &\sim& \sqrt{2 \ln \left( \frac{\eta_{\max}}{\eta}\right)},
\end{eqnarray*}
as given by equations~(\ref{eq:tailR2r2rmax}) and~(\ref{eq:deltan}). 
Equation~\eqref{eq:Lapmeth} is due to applying the Laplace approximation method~\cite{butler_2007} around $x = 1$. This concludes the proof of the lemma.
\end{proof}

We now derive Theorem~\ref{th:cherinfo} which states
\Bayesian*

\begin{proof}
Since the MAP detector is a LRT~\eqref{NPdet} for $\tilde{\eta} = \frac{\pi_C}{\pi_G} > 0$, for $n$ large enough and $\eta = \tilde{\eta}$ fixed, the detector~\eqref{NPdet} will be operating under Regime (R-i) and hence the expressions of $(1 - P_D)$ and $P_F$ of Lemma~\ref{lem:asymptotics} do apply. Replacing  $\eta_{\max}$ by its asymptotic behavior as characterized by equation~(\ref{eq:fmaxtail}) yields the required expression for $(1 - P_D)$. The same applies to $P_F$ after using the following approximation~\cite{borjesson79}: 
\begin{equation}
Q(x) \underset{x \rightarrow \infty}{\sim} \frac{1}{\sqrt{2 \pi}} \frac{e^{-\frac{x^2}{2}}}{x}.\label{eq:Qapproxneg}
\end{equation}
\end{proof}

\subsection{The Neyman-Pearson Formulation; Proof}

When it comes to the Neyman-Pearson criteria, we first estalish the following lemma:
\begin{lemma}
\label{lem:regime}
For $n$ large enough, the optimal Neyman-Pearson LRT detectors operate,
\begin{itemize}
\item under Regime (R-i) {\em whenever $P_F = \epsilon$\/},
\item under Regime (R-ii), {\em whenever $(1 - P_D) = \epsilon$\/}. 
\end{itemize}
\end{lemma}

\begin{proof}
Since $P_F$ increases with $\eta$, in order for $P_F$ to be equal to $\eps \in (0,1)$ as $n$ grows to infinity, then $\eta$ must increase "as fast as $\eta_{\max}$" with $n$ for otherwise $P_F$ would go to zero as implied by our analysis in Lemma~\ref{lem:asymptotics} and our observations in Section~\ref{sec:pre}. Moreover, since $\eta_{\max}^{\frac{2}{n+1}} \underset{n \rightarrow \infty}{\longrightarrow} 1$~\eqref{eq:fmaxtail}, and $\eta \leq \eta_{\max}$ then by inspection, one can see that necessarily $\eta^{\frac{2}{n+1}} \underset{n \rightarrow \infty}{\longrightarrow} 1$ which corresponds to Regime (R-i). 

Similarly, $(1 - P_D)$ is decreasing with $\eta$ and for a fixed $\eta$ $(1 - P_D) \underset{n \rightarrow \infty}{\longrightarrow} 0$. In order for $(1 - P_D)$ to be equal to $\eps$, $\eta$ must decrease with $n$. 
We prove next that Regime (R-iii) (and a fortiori $\eta < \eta_0$) is not possible for otherwise $(1 - P_D) \underset{n \rightarrow \infty}{\longrightarrow} 1$. In fact, equation~\eqref{eq:uppbdapp} in Appendix~\ref{app:proofeq} gives a lowerbound on $1 - P_D$:
\begin{align}
&1 - P_D \nonumber\\
&\geq \frac{1}{\sqrt{\pi}} \frac{\Gamma\!\left(\frac{n+1}{2}\right)}{  \Gamma\!\left(\frac{n}{2}\right)} \sqrt{\frac{2}{n+1}} \, \left[\Gamma\!\left(\frac{1}{2}, \frac{n+1}{2r_2^2}\right) - \Gamma\!\left(\frac{1}{2}, \frac{n+1}{2r_1^2}\right)\right] \nonumber\\
&\underset{n \rightarrow \infty}{\longrightarrow} \frac{1}{\sqrt{\pi}} \,\Gamma\!\left(\frac{1}{2},0\right) - 0= 1, \nonumber
\end{align}
where the limit result is due to Stirling's identity and to the asymptotic behavior of $r_1^2$ and $r_2^2$ under Regime (R-iii) according to equations~(\ref{tailr1R22}) and~(\ref{tailr2R22}). Furthermore, we show that $\eta$ must be under Regime (R-ii) by ruling out Regime (R-i) as well. Indeed, under case (R-i), the condition of Lemma~\ref{lem:asymptotics} is satisfied, and hence 
\begin{equation*}
1 - P_D \sim  \frac{2}{\sqrt{\pi}} \frac{\sqrt{\ln\left(\frac{\eta_{\max}}{\eta}\right)}}{\eta_{\max}} \underset{n \rightarrow \infty}{\longrightarrow} 0,
\end{equation*}
where we used the fact that $\eta_{\max} = \Theta(\sqrt{n})$~\eqref{eq:fmaxtail}. This rules out the Regime (R-i) thus proving our assertion. 
\end{proof}

We recall Theorem~\ref{th:error}:
\NP*

\begin{proof}
Let $0 < \epsilon < 1$. We begin with the proof of Case 1. Whenever $P_F = \eps$, the results of Lemma~\ref{lem:regime} imply that Regime (R-i) is the operating regime. Therefore, we use the result of Lemma~\ref{lem:asymptotics} to write the expression of $P_F$: 
 \begin{align}
& \lim_{n \rightarrow \infty}2 \, Q\left(\sqrt{2 \ln \left( \frac{\eta_{\max}}{\eta}\right)}\right) = \eps \nonumber\\
\implies& \lim_{n \rightarrow \infty} \ln \left(\frac{\eta_{\max}}{\eta}\right) = \frac{1}{2}\left(Q^{-1}\left(\frac{\epsilon}{2}\right)\right)^2 \label{eq:etavalpf},
\end{align}
where $Q^{-1}(\cdot)$ is the inverse $Q$-function. 
Using the expression of $(1 - P_D)$ as given in Lemma~\ref{lem:asymptotics}, we get: 
\begin{align}
1- P_D & \sim \frac{2}{\sqrt{\pi}} \frac{\sqrt{\ln\left(\frac{\eta_{\max}}{\eta}\right)}}{\eta_{\max}} \nonumber\\
&\sim \sqrt{\frac{2}{\pi}} \frac{2}{\xi}e^{-\frac{1}{\xi^2}}Q^{-1}\left(\frac{\epsilon}{2}\right) \frac{1}{\sqrt{n}}, \label{eq:etabeha}
\end{align}
where the last equality follows from equation~\eqref{eq:etavalpf} and equation~\eqref{eq:fmaxtail}. This completes the proof of Case 1. 

We move now to proving Case 2 of the theorem. As shown in Lemma~\ref{lem:regime}, Regime (R-ii) is the OR and therefore the expression given by Lemma~\ref{lem:asymptotics} does not hold. 
We proceed by writing $P_F$ as given by equation~\eqref{eq:changevarLap}:
\begin{align}
&P_F 
= \frac{\left(\frac{n}{2} - 1\right)^{\frac{n}{2}}}{\Gamma\!\left(\frac{n}{2}\right)} \left[ \int^{\frac{2 r_1^2}{(n-2)\xi^2}}_{0} \hspace{-3pt} e^{\left(\frac{n}{2} - 1\right)(\ln x  - x)}\,dx \right. \nonumber\\
& \qquad \qquad \qquad \qquad \qquad \left. + \int_{\frac{2 r_2^2}{(n-2)\xi^2}}^{\infty}  e^{\left(\frac{n}{2} - 1\right)(\ln x  - x)} \,dx \right] \nonumber \\
&\leq \frac{\left(\frac{n}{2} - 1\right)^{\frac{n}{2}} e^{-\left(\frac{n}{2} - 1\right)}}{\Gamma\!\left(\frac{n}{2}\right)} \left[ \int^{\frac{2 r_1^2}{(n-2)\xi^2}}_{0}  e^{-\left(\frac{n}{2} - 1\right)\frac{(x-1)^2}{2}}\,dx  \right. \nonumber \\
& \qquad \qquad \qquad \qquad \,\,  \left. + \, \int_{\frac{2 r_2^2}{(n-2)\xi^2}}^{\infty}  e^{-\left(\frac{n}{2} - 1\right) K(\epsilon) (x-1)}\,dx \right] \label{eq:upperexpo}\\
&= \frac{\left(\frac{n}{2} - 1\right)^{\frac{n}{2}} e^{-\left(\frac{n}{2} - 1\right)}}{\Gamma\!\left(\frac{n}{2}\right)} \nonumber\\
&\qquad \qquad \left[ \sqrt{\frac{2\pi}{\frac{n}{2}-1}} \left(  Q\!\left(\frac{-1}{\sqrt{\frac{2}{n-2}}}\right) - Q\!\left(\frac{\frac{2 r_1^2}{(n-2)\xi^2} - 1}{\sqrt{\frac{2}{n-2}}}\right) \right) \right. \nonumber\\
& \qquad \qquad \qquad \quad \left. + \frac{1}{K(\epsilon)\left(\frac{n}{2}-1\right)} e^{-\left(\frac{n}{2} - 1\right) K(\epsilon) (\frac{2 r_2^2}{(n-2)\xi^2}-1)} \right] \nonumber\\
&\sim \frac{1}{\sqrt{\pi}} \, \frac{1}{\sqrt{n-2}} \left[ \, \frac{e^{-\frac{(n-2)\left(1 +  \mathcal{W}_{0} \left(-e^{-1} \kappa(\epsilon)\right)\right)^2}{4}}}{1 +  \mathcal{W}_{0} \left(-e^{-1} \kappa (\epsilon)\right)} - e^{- \frac{n-2}{4}} \right.\nonumber\\
&\qquad \left. + \frac{1}{K(\epsilon) \sqrt{\pi} \sqrt{n-2}} e^{-\left(\frac{n}{2} - 1\right) K(\epsilon) \left[ -1 - \mathcal{W}_{-1} \left(-e^{-1} \kappa(\epsilon)\right)\right]} \right]\label{eq:Q}\\
&\sim \frac{\kappa_3(\epsilon)e^{- \kappa_4(\epsilon) n }}{\sqrt{n}}, \nonumber 
\end{align}
where $\kappa_3(\epsilon) > 0$, $\kappa_4(\epsilon) > 0$ can be judiciously chosen function of the different parameters in equation~(\ref{eq:Q}). Equation~\eqref{eq:Q} is due to Stirling's approximation, to the approximation in~\eqref{eq:Qapproxneg}, and to the fact that 
\begin{align*}
\lim_{n \rightarrow \infty} \frac{2 r_1^2}{(n-2)\xi^2}  \, =& - \mathcal{W}_{0} \left(-e^{-1} \kappa(\epsilon) \right) < 1\\ 
\& \quad \lim_{n \rightarrow \infty} \frac{2 r_2^2}{(n-2)\xi^2}  \, =& - \mathcal{W}_{-1} \left(-e^{-1} \kappa(\epsilon) \right) > 1,
\end{align*}
given the asymptotic behavior of $r_1^2$ and $r_2^2$ under case (R-ii) as given by equations~(\ref{eq:caseii0}) and (\ref{eq:caseii-1}), and where $\kappa(\epsilon) \eqdef \lim_{n \rightarrow \infty} \eta^{\frac{2}{n+1}}$, $0 < \kappa(\epsilon) <1$. Equation~\eqref{eq:upperexpo} is due to the following identities:  
\begin{equation*}
\ln x - x \leq -1 - \frac{(x-1)^2}{2} \qquad 0 \leq x \leq 1,
\end{equation*}
and for every $c > 1$ there exits a $K \in(0, 1)$ such that: 
\begin{equation*}
\ln x - x \leq -1 - K (x - 1) \qquad  c \leq x .
\end{equation*}
Finally, we provide a lowerbound on $P_F$:
\begin{align}
P_F&= \frac{1}{\Gamma\!\left(\frac{n}{2}\right)} \!\left[ \int^{\frac{r_1^2}{\xi^2}}_{0}  \hspace{-3pt} r^{\frac{n}{2} - 1} e^{-r}\,dr + \int_{\frac{r_2^2}{\xi^2}}^{\infty} \hspace{-3pt} r^{\frac{n}{2} - 1} e^{-r}\,dr \right] \nonumber\\
&\geq \frac{1}{\Gamma\!\left(\frac{n}{2}\right)} \int_{\frac{r_2^2}{\xi^2}}^{\infty}  r^{\frac{n}{2} - 1} e^{-r}\,dr  \nonumber\\
&= \frac{\left(\frac{n}{2} - 1\right)^{\frac{n}{2}}}{\Gamma\!\left(\frac{n}{2}\right)} \int_{\frac{2 r_2^2}{(n-2)\xi^2}}^{\infty}  e^{\left(\frac{n}{2} - 1\right)(\ln x  - x)}\,dx \nonumber\\
&\geq \frac{\left(\frac{n}{2} - 1\right)^{\frac{n}{2}}  e^{-\left(\frac{n}{2} - 1\right)}}{\Gamma\!\left(\frac{n}{2}\right)} \hspace{-0.05cm}  \int_{\frac{2 r_2^2}{(n-2)\xi^2}}^{\infty}  \hspace{-0.05cm} e^{-\left(\frac{n}{2} - 1\right)\frac{(x  - 1)^2}{2}} dx \label{eq:lowerbdhigh}\\
&= \frac{\left(\frac{n}{2} - 1\right)^{\frac{n}{2}}  e^{-\left(\frac{n}{2} - 1\right)}}{\Gamma\!\left(\frac{n}{2}\right)}   \sqrt{\frac{2\pi}{\frac{n}{2}-1}} Q\!\left(\frac{\frac{2 r_2^2}{(n-2)\xi^2} - 1}{\sqrt{\frac{2}{n-2}}}\right)\nonumber\\
&\sim \frac{1}{\sqrt{\pi}} \, \frac{1}{\sqrt{n-2}} \, \frac{e^{-\frac{(n-2)\left[-1 - \mathcal{W}_{-1} \left(-e^{-1} \kappa(\epsilon)\right)\right]^2}{4}}}{-1 -  \mathcal{W}_{-1} \left(-e^{-1} \kappa (\epsilon)\right)} \label{eq:caseiiibd}\\ 
&= \frac{\kappa_1(\epsilon)e^{- \kappa_2(\epsilon) n }}{\sqrt{n}}, \nonumber 
\end{align}
for some $\kappa_1(\epsilon), \kappa_2(\epsilon) > 0$. Equation~\eqref{eq:lowerbdhigh} is due to the fact that $\ln x - x \geq -1 - \frac{(x - 1)^2}{2} $ for $x \geq 1$ and equation~\eqref{eq:caseiiibd} is justified in a similar manner as done for equation~(\ref{eq:Q}). 
\end{proof}

\section{Conclusion}
\label{sec:con}
It is well-known that, in a discrete setup and under IID observations, the error of the second kind $(1 - P_D)$ goes to zero when $n \rightarrow \infty$ according to $2^{-nD(p\|q)}$, where 
$p$ and $q$ are the probability laws under the two different hypotheses. In some cases, a similar observation can be inferred whenever $p$ and $q$ are continuous~\cite{tusnady1977asymptotically}. Though the results presented in this paper are for a specific problem on hypothesis testing between two isotropic probability laws, they are rather interesting on their own and are ``non-standard" when compared to the IID ``Gaussian"-centric case. Our findings suggest new insights on the hypothesis testing problem: 
\begin{enumerate}
\item As a generalization to the IID case, the leading term in the exponent might be possibly governed by the behavior of $D(p_n\|q_n)$ as a function of $n$ whenever the relative entropy is finite. This extension is currently being investigated by the authors. Note that, under the IID case, we have $D(p_n\|q_n) = n D(p\|q)$ which recovers the standard result of the IID setup.  
\item In general $D(p_n\|q_n)$ is not linear in $n$. When $p_n$ is \text{circular (IID)} $\mathcal{N}(0, \sigma^2)$ and $q_n$ is $\text{circular} \,\, \mathcal{C}(0,\gamma)$, we show in Appendix~\ref{app:KL} that $D(p_n\|q_n) = \Theta\left(\ln \sqrt{n}\right)$ which corroborates the previous point as per the results of Theorem~\ref{th:error} for $(1 - P_D)$. 
\item Whenever $D(p_n\|q_n)$ is not finite, the term in the exponent is not linear: It was found in this setup to be at least $\Theta(n) + \frac{1}{2} \ln n$. 
\end{enumerate}

\bibliographystyle{IEEEtran}
\bibliography{paper}

\begin{thebibliography}{10}
\providecommand{\url}[1]{#1}
\csname url@rmstyle\endcsname
\providecommand{\newblock}{\relax}
\providecommand{\bibinfo}[2]{#2}
\providecommand\BIBentrySTDinterwordspacing{\spaceskip=0pt\relax}
\providecommand\BIBentryALTinterwordstretchfactor{4}
\providecommand\BIBentryALTinterwordspacing{\spaceskip=\fontdimen2\font plus
\BIBentryALTinterwordstretchfactor\fontdimen3\font minus \fontdimen4\font\relax}
\providecommand\BIBforeignlanguage[2]{{%
\expandafter\ifx\csname l@#1\endcsname\relax
\typeout{** WARNING: IEEEtran.bst: No hyphenation pattern has been}%
\typeout{** loaded for the language `#1'. Using the pattern for}%
\typeout{** the default language instead.}%
\else
\language=\csname l@#1\endcsname
\fi
#2}}

\bibitem{cover}
{T. M. Cover and J. A. Thomas}, \emph{{Elements of Information Theory}}, 2nd~ed.\hskip 1em plus 0.5em minus 0.4em\relax John Wiley \& Sons, 2006.

\bibitem{tusnady1977asymptotically}
G.~Tusn{\'a}dy, ``On asymptotically optimal tests,'' \emph{The Annals of Statistics}, pp. 385--393, 1977.

\bibitem{csiszar1998}
I.~Csiszar, ``The method of types [information theory],'' \emph{IEEE Transactions on Information Theory}, vol.~44, no.~6, pp. 2505--2523, 1998.

\bibitem{kolmo}
{B. V. Gnedenko and A. N. Kolmogorov}, \emph{Limit Distributions for Sums of Independent Random Variables}.\hskip 1em plus 0.5em minus 0.4em\relax Reading Massachusetts: Addison-Wesley Publishing Company, 1968.

\bibitem{nbrm2001}
\BIBentryALTinterwordspacing
J.~P. Nolan, \emph{Maximum Likelihood Estimation and Diagnostics for Stable Distributions}.\hskip 1em plus 0.5em minus 0.4em\relax Boston, MA: Birkh{\"a}user Boston, 2001, pp. 379--400. [Online]. Available: \url{https://doi.org/10.1007/978-1-4612-0197-7\_17}
\BIBentrySTDinterwordspacing

\bibitem{kuzer2004}
E.~Kuruoglu and J.~Zerubia, ``{Modeling SAR images with a generalization of the Rayleigh distribution},'' \emph{IEEE Transactions on Image Processing}, vol.~13, no.~4, pp. 527--533, 2004.

\bibitem{nolancirc2013}
\BIBentryALTinterwordspacing
J.~P. Nolan, ``Multivariate elliptically contoured stable distributions: theory and estimation,'' \emph{Computational Statistics}, vol.~28, no.~5, pp. 2067--2089, 2013. [Online]. Available: \url{https://doi.org/10.1007/s00180-013-0396-7}
\BIBentrySTDinterwordspacing

\bibitem{corless1996}
\BIBentryALTinterwordspacing
R.~M. Corless, G.~H. Gonnet, D.~E.~G. Hare, D.~J. Jeffrey, and D.~E. Knuth, ``{On the LambertW function},'' \emph{Advances in Computational Mathematics}, vol.~5, no.~1, pp. 329--359, 1996. [Online]. Available: \url{https://doi.org/10.1007/BF02124750}
\BIBentrySTDinterwordspacing

\bibitem{butler_2007}
R.~W. Butler, \emph{Saddlepoint Approximations with Applications}, ser. Cambridge Series in Statistical and Probabilistic Mathematics.\hskip 1em plus 0.5em minus 0.4em\relax Cambridge University Press, 2007.

\bibitem{borjesson79}
P.~Borjesson and C.-E. Sundberg, ``{Simple Approximations of the Error Function Q(x) for Communications Applications},'' \emph{IEEE Transactions on Communications}, vol.~27, no.~3, pp. 639--643, 1979.

\bibitem{Hisakado06}
M.~Hisakado, K.~Kitsukawa, and S.~Mori, ``Correlated binomial models and correlation structures,'' \emph{Journal of Physics A General Physics}, vol.~39, 06 2006.

\bibitem{Nicola1990}
V.~Nicola and A.~Goyal, ``Modeling of correlated failures and community error recovery in multiversion software,'' \emph{IEEE Transactions on Software Engineering}, vol.~16, no.~3, pp. 350--359, 1990.

\bibitem{spencer2014}
J.~H. Spencer and L.~Florescu, \emph{Asymptopia}.\hskip 1em plus 0.5em minus 0.4em\relax Providence, Rhode Island: American Mathematical Society, 2014.

\end{thebibliography}
\appendices
\section{The KL Divergence}
\label{app:KL}
We characterize in this appendix the behavior at large values of $n$ of the relative entropy between an IID $n$-dimensional zero-mean and variance $\frac{\xi^2}{2}$ Gaussian distribution $p_G \eqdef \mathcal{N}\left(0,\frac{\xi^2}{2}\right)$ and an isotropic $n$-dimensional Cauchy
distribution $p_C$ with a scale parameter $\gamma = 1$. Our main result shows that $D(p_G\|p_C) \sim \frac{1}{2} \ln n$. However, before proceeding with the proof, we highlight the important fact that $D(p_C\|p_G)$ is infinite. Since both $p_G({\bf x})$ and $p_C({\bf x})$ only depend on $r = \|{\bf x}\|$, we abuse the notation in what follows and write $p_G(r)$ and $p_C(r)$ instead to highlight this dependency without renaming the functions. 
\begin{align}
&D(p_C\|p_G)
= - h(p_C) - \int_{\mathbb{R}^n} p_C({\bf x}) \ln p_G({\bf x}) \, d{\bf x}\nonumber\\
&= - h(p_C) + \ln\left(\pi^{\frac{n}{2}} \xi^n\right) + \frac{1}{\xi^2} \int_{\mathbb{R}^n} r^2 p_C(r)\,d{\bf x} \nonumber\\
&= - h(p_C) + \ln\left(\pi^{\frac{n}{2}} \xi^n\right) + \frac{1}{\xi^2} \frac{2 \pi^{\frac{n}{2}}}{\Gamma\!\left(\frac{n}{2}\right)}  \int_{0}^{\infty} r^{n+1} p_C(r)\,dr \nonumber\\
&= \infty \nonumber,
\end{align}
where in order to write the last equality we used the fact that $h(p_C) < \infty$ and that $ p_C(r) = \Theta\left(r^{-n - 1}\right)$ as $r \rightarrow +\infty$ as given by equation~\eqref{eq:cauchypdfn}.

Back to computing $D(p_G\|p_C)$, we assume $n \geq 2$ and we write: 
\begin{align}
&D(p_G\|p_C) = - h(p_G) - \int_{\mathbb{R}^n} p_G({\bf x}) \ln p_C({\bf x}) \, d{\bf x} \nonumber\\
&= -\frac{n}{2} \ln (\pi e \xi^2) - \ln \frac{\Gamma\left(\frac{n+1}{2}\right)}{\pi^{\frac{n+1}{2}}} \nonumber\\
&\qquad \qquad + \frac{n + 1}{2 \pi^{\frac{n}{2}} \xi^n} \int_{\mathbb{R}^n} e^{-\frac{r^2}{\xi^2}} \ln\left(1 + r^2\right)\,d{\bf x} \nonumber\\
&= -\frac{n}{2} \ln (\pi e \xi^2) - \ln \frac{\Gamma\left(\frac{n+1}{2}\right)}{\pi^{\frac{n+1}{2}}} \nonumber\\
& \qquad \qquad + \frac{n + 1}{\xi^n \Gamma\left(\frac{n}{2}\right)} \int_{0}^{\infty} e^{-\frac{r^2}{\xi^2}} \ln\left(1 + r^2\right) r^{n - 1}\,dr \nonumber\\
& =  -\frac{n}{2} \ln (\pi e \xi^2) - \ln \frac{\Gamma\left(\frac{n+1}{2}\right)}{\pi^{\frac{n+1}{2}}} +  \frac{n + 1}{2}\, I_n, \label{eq:dpnqn}
\end{align}
where 
\begin{align}
& I_n = \frac{1}{\Gamma\left(\frac{n}{2}\right)} \int_{0}^{\infty} e^{-r} \ln\left(1 + \xi^2 r\right) r^{\frac{n}{2} - 1}\,dr \nonumber\\
&= \frac{\left(\frac{n}{2}-1\right)^{\frac{n}{2}}}{\Gamma\left(\frac{n}{2}\right)}  \int_{0}^{\infty} e^{\left(\frac{n}{2}-1\right)(\ln r - r)} \ln \left[1 + \left(\frac{n}{2} - 1\right)\xi^2 r\right] dr \label{eq:integ1}\\
&= I_{1n} + I_{2n} + I_{3n}, \nonumber
\end{align}
where $I_{1n}$, $I_{2n}$, $I_{3n}$ are a partition of $I_n$ on the respective disjoint sets: $B_1 = (0,1-\epsilon_n) \cup (1 + \epsilon_n, R]$, $B_2 = [1-\epsilon_n, 1+\epsilon_n]$, and $B_3 = (R, \infty)$ for some $R > 0$ and $\epsilon_n = \left(\frac{\ln n}{\sqrt{n}}\right)$, Note that $\epsilon_n \rightarrow 0$ as $n \rightarrow \infty$ and $R$ is a constant to be chosen large enough as determined hereafter.

We first consider $I_{1n}$ and show that it is $o\left(\frac{\ln n}{n}\right)$. For that matter, it is enough to show that the integrand is $o\left(\frac{\ln n}{n}\right)$. Indeed, for $r \in B_1$, we have: 
\begin{align}
&\frac{\left(\frac{n}{2}-1\right)^{\frac{n}{2}}}{\Gamma\left(\frac{n}{2}\right)} e^{\left(\frac{n}{2}-1\right)(\ln r - r)} \ln \left[1 + \left(\frac{n}{2} - 1\right)\xi^2 r\right] \nonumber\\
\leq & \, \frac{\left(\frac{n}{2}-1\right)^{\frac{n}{2}}}{\Gamma\left(\frac{n}{2}\right)} e^{\left(\frac{n}{2}-1\right)\left[ -1  -\frac{\epsilon_n^2}{3}  \right]} \ln \left[ 1 + \left(\frac{n}{2} - 1\right)\xi^2 r \right] \label{eq:stir100b}\\
\sim & \, \frac{\left(\frac{n}{2}-1\right)^{\frac{1}{2}}}{\sqrt{2 \pi}} e^{\left(\frac{n}{2}-1\right)\left[ -\frac{\epsilon_n^2}{3} \right]} \ln \left[ \left(\frac{n}{2} - 1\right)\xi^2 r \right] \label{eq:stir100}\\
= & \, o\left(\frac{\ln n}{n}\right), \label{eq:i1n}
\end{align}
where we used the fact that $\ln r - r + 1 < \ln (1 + \epsilon_n) - \epsilon_n < -\frac{\epsilon_n^2}{3}$ on $B_1$ to write~(\ref{eq:stir100b}), and  Stirling's approximation of the Gamma function to write~(\ref{eq:stir100}).

Regarding $I_{3n} > 0$, we have: 
\begin{align}
&I_{3n} \nonumber\\
&= \frac{\left(\frac{n}{2}-1\right)^{\frac{n}{2}}}{\Gamma\left(\frac{n}{2}\right)}  \int_{R}^{\infty} e^{\left(\frac{n}{2}-1\right)(\ln r - r)} \ln \left[ 1 + \left(\frac{n}{2} - 1\right)\xi^2 r \right] \,dr \nonumber\\
&\leq \frac{\left(\frac{n}{2}-1\right)^{\frac{n}{2}}}{\Gamma\left(\frac{n}{2}\right)}  \int_{R}^{\infty} e^{-\left(\frac{n}{2}-1\right)\frac{r}{2}} \left[\ln \left[\left(\frac{n}{2} - 1\right)\xi^2\right] + r\right] \, dr \label{eq:largeR1} \\
&= \frac{\left(\frac{n}{2}-1\right)^{\frac{n}{2}}}{\Gamma\left(\frac{n}{2}\right)} \left[\ln \left[\left(\frac{n}{2} - 1\right)\xi^2\right] \frac{2}{\frac{n}{2} - 1} e^{-(\frac{n}{2}-1)\frac{R}{2}}  \right. \nonumber\\
&\left. \qquad \qquad + \frac{2}{\frac{n}{2} - 1} R \, e^{-(\frac{n}{2}-1)\frac{R}{2}} + \left(\frac{2}{\frac{n}{2} - 1}\right)^2 e^{-(\frac{n}{2}-1)\frac{R}{2}}\right] \label{eq:intparts}\\
&\sim  \frac{\left(\frac{n}{2}-1\right)^{\frac{1}{2}}}{\sqrt{2 \pi}} e^{(\frac{n}{2} - 1)}  e^{-(\frac{n}{2}-1)\frac{R}{2}}\left[\ln \left[\left(\frac{n}{2} - 1\right)\xi^2\right] \frac{2}{\frac{n}{2} - 1}  \right. \nonumber\\
&\left. \qquad \qquad + \frac{2}{\frac{n}{2} - 1} R \,  + \left(\frac{2}{\frac{n}{2} - 1}\right)^2 \right] \nonumber\\
&= o\left(\frac{\ln n}{n}\right), \label{eq:i3n}
\end{align}
which is true for $R > 2$. Equation~(\ref{eq:largeR1}) is justified by the fact that $1 + \left(\frac{n}{2} - 1\right)\xi^2 r < \left(\frac{n}{2} - 1\right)\xi^2 e^{r}$, for $n > 2$ and $r \in B_3$ with $R$ large enough. Equation~(\ref{eq:intparts}) is due to integration by parts. 

Finally, considering $I_{2n}$, we get: 
\begin{align}
&I_{2n} \nonumber\\
&= \frac{\left(\frac{n}{2}-1\right)^{\frac{n}{2}}}{\Gamma\left(\frac{n}{2}\right)}  \int_{1-\epsilon_n}^{1+\epsilon_n} \hspace{-0.15cm}  e^{\left(\frac{n}{2}-1\right)(\ln r - r)} \ln \left[1 + \left(\frac{n}{2} - 1\right)\xi^2 r\right]\,dr \nonumber\\
&\geq \left[\frac{\left(\frac{n}{2}-1\right)^{\frac{n}{2}}}{\Gamma\left(\frac{n}{2}\right)} \int_{1-\epsilon_n}^{1+\epsilon_n} e^{\left(\frac{n}{2}-1\right)(\ln r - r)} \,dr\right] \nonumber\\
&\qquad \qquad \qquad \times  \ln \left[1 + \left(\frac{n}{2} - 1\right)\xi^2 (1 - \epsilon_n)\right] \label{eq:lowerval} \\
& \geq \hspace{-0.1cm}  \left[\frac{\left(\frac{n}{2}-1\right)^{\frac{n}{2}} e^{-\left(\frac{n}{2}-1\right)}}{\Gamma\left(\frac{n}{2}\right)}  \hspace{-0.14cm} \int_{1-\epsilon_n}^{1+\epsilon_n} \hspace{-0.3cm} e^{-\left(\frac{n}{2}-1\right)\left[\frac{1}{2}(r - 1)^2 + \frac{2|r-1|^3}{3}\right]} dr\right]  \nonumber \\
&\qquad \qquad \qquad \times \ln \left[1 + \left(\frac{n}{2} - 1\right)\xi^2 (1 - \epsilon_n)\right] \label{eq:taylorseries} \\
&\geq \frac{\left(\frac{n}{2}-1\right)^{\frac{n}{2}} e^{-\left(\frac{n}{2}-1\right)}}{\Gamma\left(\frac{n}{2}\right)}  \sqrt{\frac{2\pi}{\frac{n}{2}-1}} e^{-\left(\frac{n}{2}-1\right) \frac{2 \epsilon_n^3}{3}}  \nonumber\\
& \left[1 - 2 Q\left(\epsilon_n \sqrt{\frac{n}{2} - 1} \right)\right]  \times \ln \left[1 + \left(\frac{n}{2} - 1\right)\xi^2 (1 - \epsilon_n)\right] \label{eq:errorterm} \\
&\sim \ln \left(\frac{n}{2}\, \xi^2 \,\right) \label{eq:stir101},
\end{align}
where again we used Stirling's approximation of the Gamma function in order to write equation~(\ref{eq:stir101}). Equation~\eqref{eq:taylorseries} is due to the fact that $\ln(r) - r + 1 \geq -(r-1)^2/2 - 2|r-1|^3/3$ around $r = 1$ and equation~\eqref{eq:errorterm} is justified by the fact that the term $|r - 1|^3 \leq \epsilon_n^3$ whenever $|r - 1| \leq \epsilon_n$. 
A similar upperbound can be readily obtained by replacing $\ln \left[1 + \left(\frac{n}{2} - 1\right)\xi^2 (1 - \epsilon_n)\right]$ by $\ln \left[1 + \left(\frac{n}{2} - 1\right)\xi^2 (1 + \epsilon_n)\right]$ in equation~\eqref{eq:lowerval}; upperbounding $\ln r - r + 1 \leq -(r-1)^2/2 + |r-1|^3/3$; and
by replacing $e^{-\left(\frac{n}{2}-1\right) \frac{2\epsilon_n^3}{3}}$ by $e^{\left(\frac{n}{2}-1\right) \frac{\epsilon_n^3}{3}}$ in equation~\eqref{eq:errorterm}. This implies that $I_{2n} \sim \ln \left( \frac{n}{2}\, \xi^2 \right)$ which in turn implies by virtue of equations~(\ref{eq:i1n}) and~(\ref{eq:i3n}) that 
\begin{equation}
\label{eq:in}
\frac{n + 1}{2} \, I_{n} \sim \frac{n + 1}{2} \ln \left(\frac{n}{2}\, \xi^2 \,\right) + o(\ln n).
\end{equation}

Using Stirling's approximation in addition to~\eqref{eq:in}, equation~(\ref{eq:dpnqn}) gives
\begin{align*}
&D(p_G\|p_C) = -\frac{n}{2} \ln (\pi e \xi^2) - \ln \frac{\Gamma\left(\frac{n+1}{2}\right)}{\pi^{\frac{n+1}{2}}} +  \frac{n + 1}{2}\, I_n \\
&\sim -\frac{n}{2} \ln (\pi e \xi^2) - \frac{n - 1}{2}\ln \left(\frac{n - 1}{2}\right) + \frac{n - 1}{2} \\
& - \frac{1}{2}\ln \left(\frac{n - 1}{2}\right) + \frac{n + 1}{2} \ln \pi  +\frac{n + 1}{2} \ln \left(\frac{n}{2}\, \xi^2 \,\right) + o(\ln n) \\
&\sim \frac{1}{2} \ln n.
\end{align*}

\section{Correlated Hypotheses; Examples}
\label{app:corr}

\subsection{Example 1}

Let $0 < \alpha < 1$. Consider the following binary hypothesis testing problem
\begin{eqnarray*}
     \left\{ \begin{array}{ll}
        \displaystyle  H_0: & \{X_i\}_{i=1}^{n} \text{ IID } \sim \mathcal{B}\left(\frac{1}{2}\right) \vspace{5pt}  \\
        \displaystyle  H_1: &  \begin{cases} X_1 \sim \mathcal{B}\left(\frac{1}{2}\right) \,\, \& \,\, X_1 = \cdots = X_n & \text{w.p.} \,\, \alpha\\ \{X_i\}_{i=1}^{n} \text{ IID } \sim \mathcal{B}\left(\frac{1}{2}\right) & \text{w.p.} \,\, (1- \alpha),\end{cases}
      \end{array} \right.
  \end{eqnarray*}
where $\mathcal{B}(p)$ is the Bernoulli variable with parameter $p$. It can be readily dervied that the Maximum Likelihood (ML) device --optimal in a Bayesian setting with equal a priori's-- decides on $H_1$ whenever all $X_i$'s are equal and on $H_0$ otherwise. Computing $P_F$: 
\begin{align*}
P_F = \, & \text{Pr}\left(\hat{H} = H_1| H_0\right) \nonumber\\
= \, & \text{Pr}\left((\underset{n \,\, \text{times}}{\underbrace{0,0, \cdots,0}}) \cup (\underset{n \,\, \text{times}}{\underbrace{1,1, \cdots,1}}) | H_0\right) = \left(\frac{1}{2}\right)^{n-1},
\end{align*}
which is exponentially decaying in $n$. As for $P_D$: 
\begin{align*}
P_D = \, & \text{Pr}\left(\hat{H} = H_1| H_1\right) \nonumber\\ 
= \, & \text{Pr}\left((\underset{n \,\, \text{times}}{\underbrace{0,0, \cdots,0}}) \cup (\underset{n \,\, \text{times}}{\underbrace{1,1, \cdots,1}}) | H_1\right)\\
= \, & 2 \left(\frac{\alpha}{2} + (1 - \alpha)\left(\frac{1}{2}\right)^n\right), 
\end{align*}
which implies that 
\begin{equation*}
1 - P_D = (1 - \alpha) \left(1 - \left(\frac{1}{2}\right)^{n-1}\right),
\end{equation*}
converging to a positive constant as $n$ goes to infinity and is clearly different in behavior from $P_F$ at large values of $n$.

\subsection{Example 2: The Beta-Binomial Distribution}
Let 
 
\begin{eqnarray*}
     \left\{ \begin{array}{ll}
        \displaystyle  H_0: & \{X_i\}_{i=1}^{n} \text{ IID } \sim \mathcal{B}\left(\frac{1}{2}\right) \vspace{5pt} \\ 
        \displaystyle  H_1: &  \mathcal{BB}(n,1,1), \quad \text{correlated} 
      \end{array} \right.
  \end{eqnarray*}
  where $\mathcal{BB}(n,\alpha,\beta)$ is the Beta-Binomial distribution with parameters $n \in \Naturals$, and $\alpha, \beta > 0$~\cite{Hisakado06,Nicola1990}. Under the $\mathcal{BB}(n,\alpha,\beta)$ distribution, the probability of having $i$ zeros among length-$n$ sequences is equal to: 
  \begin{equation}
  \label{eq:bbpmf}
  b_n(i) = \begin{pmatrix} n \\ i \end{pmatrix} \frac{B(\alpha + i, n + \beta - i)}{B(\alpha,\beta)},
  \end{equation}
  where $B(\alpha,\beta)$ is the Beta function given in terms of the Gamma function as follows: 
  \begin{equation*}
  B(\alpha,\beta) = \frac{\Gamma(\alpha)\Gamma(\beta)}{\Gamma(\alpha + \beta)}. 
\end{equation*}
In out case $\alpha = \beta = 1$, and $b_n(i)$ evaluates to: 
\begin{align}
b_n(i) & = \begin{pmatrix} n \\ i \end{pmatrix} \frac{B(1 + i, n + 1 - i)}{B(1,1)} = \begin{pmatrix} n \\ i \end{pmatrix} \frac{i! \, (n - i)!}{(n + 1)!} \nonumber\\
&= \frac{1}{n + 1}, \quad 0 \leq i \leq n. \label{eq:uni}
\end{align}
  Under equal apriori probabilities $p_0 = p_1 = \frac{1}{2}$, the MAP detector is a ML one and will choose hypothesis $H_1$ whenever the length $n$ sequence is more probable under $H_1$ than under $H_0$. Let $i$ be the number of zeros in a length-$n$ sequence. Then, $\hat{H} = H_1$ whenever: 
  \begin{align}
  &\frac{B\left(1 + i, n + 1 - i\right)}{B\left(1,1\right)} \geq \left(\frac{1}{2}\right)^n \nonumber\\
  \iff & \frac{i! \, (n - i)!}{(n + 1)!} \geq \left(\frac{1}{2}\right)^n \nonumber\\
  \iff & \begin{pmatrix} n \\ i \end{pmatrix} \leq \frac{2^n}{n + 1}\label{eq:condition}.
  \end{align}
We note that whenever that there exists an index $i$ satisfying equation~(\ref{eq:condition}), then $(n - i)$ will also satisfy~(\ref{eq:condition}). For the range $\lceil \frac{n}{2} \rceil \leq i \leq n$, let $i_0(n)$ be the smallest integer for which~(\ref{eq:condition}) is satisfied. For large $n$, we claim that: 
\begin{equation}
\frac{i_0(n) - \frac{n}{2}}{\sqrt{n}} \sim  c \, \sqrt{\ln n}
\label{eq:limitindex}
\end{equation}
for $c = \frac{1}{2}$. Define $i^*(n) = \left \lceil \frac{n}{2} + c \sqrt{n \ln n} \right \rceil$. Since, for any $c > 0$, we have $|\frac{n}{2} - i^*(n)| = o(n^{\frac{2}{3}})$, then using Stirling's approximation for the binomial coefficient for large $n$ and large $i$~\cite{spencer2014}, we have: 
\begin{align}
\begin{pmatrix} n \\ i^*(n) \end{pmatrix} &\sim \frac{2^n}{\sqrt{\frac{\pi}{2}\,n}} e^{- \frac{(n - 2i^*(n))^2}{2n}} \nonumber\\
&\sim \sqrt{\frac{2}{\pi}}\frac{2^n}{n^{\frac{1}{2} + 2c^2}} \label{eq:bdlimindex}.
\end{align}
Equation~(\ref{eq:bdlimindex}) violates condition~(\ref{eq:condition}) whenever $c < \frac{1}{2}$ and satisfies it whenever $c \geq \frac{1}{2}$ thus proving our claim. For large $n$, we have
  \begin{align}
&P_F = \text{Pr}\left(\hat{H} = H_1| H_0\right) \nonumber\\
& =  \left(\frac{1}{2}\right)^n 2 \sum_{i_0(n)}^n \begin{pmatrix} n \\ i \end{pmatrix} \nonumber\\
& = \left(\frac{1}{2}\right)^n \left(2^n -  \sum_{n - i_0(n) + 1}^{i_0(n)-1}\begin{pmatrix} n \\ i \end{pmatrix}\right) \nonumber\\
& \sim 1 - \sqrt{\frac{2}{\pi}} \frac{1}{\sqrt{n}} \sum_{n - i_0(n) + 1}^{i_0(n)-1} e^{-\frac{(n - 2i)^2}{2n}}  \label{eq:approx11}\\
& \sim 1   - \frac{2}{\sqrt{n}}\frac{1}{\sqrt{2 \pi}} \int_{n - i_0(n) + 1}^{i_0(n)-1} e^{-\frac{(n - 2t)^2}{2n}} \, dt \label{eq:serint}\\
&= 1  -  \left[Q\left(\frac{2i_0(n) - n - 2}{\sqrt{n}}\right) - Q\left(\frac{- 2i_0(n) + n  + 2}{\sqrt{n}}\right)\right] \label{eq:chov}\\
&= 2 Q\left(\frac{2i_0(n) - n - 2}{\sqrt{n}}\right) \nonumber\\
&\sim \sqrt{\frac{2}{\pi}} \frac{e^{-\frac{(2i_0(n) - n - 2)^2}{2 n}}}{\frac{2i_0(n) - n - 2}{\sqrt{n}}} \label{eq:Qapproxappendix}\\
&\sim \sqrt{\frac{2}{\pi}} \frac{1}{\sqrt{n \ln n}} \nonumber 
\end{align}
where we used Stirling's approximation to write equation~(\ref{eq:approx11}). Equation~(\ref{eq:serint}) is valid since $e^{-\frac{(n - 2t)^2}{2n}} \sim \frac{1}{\sqrt{n}}$ is convex for $n - i_0(n) + 1 \leq t \leq i_0(n)-1$, thus implying: 
\begin{align*}
 &\int_{n - i_0(n) + 1}^{i_0(n)-1} e^{-\frac{(n - 2t)^2}{2n}} \, dt \\
 &\leq \sum_{n - i_0(n) + 1}^{i_0(n)-1} e^{-\frac{(n - 2i)^2}{2n}} \leq \int_{n - i_0(n)}^{i_0(n)-1} e^{-\frac{(n - 2t)^2}{2n}} \, dt,
\end{align*}
with both the upper and lower bounds having the same asymptotic behavior. Equation~(\ref{eq:chov}) is due to the change of variable $u = \frac{2t - n}{\sqrt{n}}$ and equation~(\ref{eq:Qapproxappendix}) is justified by using the approximation in~(\ref{eq:Qapproxneg}). The last equation is justified by equation~(\ref{eq:limitindex}). As for $1 - P_D$, we have
\begin{eqnarray}
1 - P_D &=& \text{Pr}\left(\hat{H} = H_0| H_1\right) \nonumber\\
&=&  \sum_{n - i_0(n) + 1}^{i_0(n)-1} \, b_n(i)\nonumber\\
&=& \frac{1}{n + 1} \, (2i_0(n) - n - 2), \nonumber\\
&\sim& \sqrt{\frac{\ln n}{n}}
\end{eqnarray}
where $b_n(i)$ is given in equation~(\ref{eq:uni}) and where we used equation~(\ref{eq:limitindex}) to write the last equation. Once more, $P_F$ and $1 - P_D$ exhibit different asymptotic behaviors. 



\subsection{Example 3: Continuous Laws}
Given $n \in \mathbb{N}$, $n \geq 3$, consider the following PDF over $[0,1]^n$:
\begin{align*} 
    p_n(x^n) = \begin{cases}
    n - \sqrt{n}, & \text{if } \| x^n \|_\infty \leq n^{-1/n}, \vspace{5pt} \\
    \frac{\sqrt{n}}{n-1}, & \text{otherwise},
    \end{cases}
\end{align*}
where $\|x^n \|_\infty = \max_i |x_i|$. It can be easily verified that $p_n$ is a valid PDF by noting that the volume of the set $\{x^n: \| x^n \|_\infty \leq n^{-1/n} \}$ is $\frac{1}{n}$.

Now, consider the following hypothesis testing problem:
\begin{eqnarray*}
     \left\{ \begin{array}{ll}
        \displaystyle  H_0: \{X\}_{i=1}^n \text{ IID } \sim \mathcal{U}[0,1], \\ 
        \displaystyle  H_1: X^n \sim p_n(\cdot),
      \end{array} \right.
  \end{eqnarray*}
  where $\mathcal{U}[0,1]$ is the uniform law on the interval $[0,1]$.

  Assuming a prior uniform over $H_0$ and $H_1$, the optimal ML decision rule yields
  \begin{align*}
       \| x^n \|_\infty  \underset{H_0}{\overset{H_1}{\lesseqgtr}}  n^{-1/n}.
  \end{align*}
As such,
\begin{align*}
    P_F & = \Pr ( \hat{H} = H_1 | H_0)
    = \Pr \Big(  \| X^n \|_\infty \leq n^{-1/n} | H_0 \Big) \\
    & = \frac{1}{n}.
\end{align*}
On the other hand,
\begin{align*}
    1- P_D & = \Pr ( \hat{H} = H_0 | H_1) \\
    & = \Pr (  \| X^n \|_\infty > n^{-1/n} | H_1) \\
    & = \int_{\{\|x^n\|_\infty > n^{-1/n}\}} \frac{\sqrt{n}}{n-1} \, d x^n \\
    & = \frac{\sqrt{n}}{n-1} \left( 1- \int_{\{\|x^n\|_\infty \leq n^{-1/n}\}} dx^n \right) \\
    & = \frac{\sqrt{n}}{n-1} \left( 1- \frac{1}{n} \right) \\
    & = \frac{1}{\sqrt{n}},
\end{align*}
and the asymptotic behavior is clearly different from that of $P_F$.

\section{$(1 - P_D)$}
\label{app:proofeq}
In this appendix, we show the following technical result. We have 
\begin{align}
&1- P_D \nonumber\\
&\geq \frac{1}{\sqrt{\pi}} \frac{\Gamma\!\left(\frac{n+1}{2}\right)}{  \Gamma\!\left(\frac{n}{2}\right)} \sqrt{\frac{2}{n+1}} \, \left(\Gamma\!\left(\frac{1}{2}, \frac{n+1}{2r_2^2}\right) \hspace{-0.003cm}- \hspace{-0.003cm} \Gamma\!\left(\frac{1}{2}, \frac{n+1}{2r_1^2}\right)\right).  \label{eq:uppbdapp} 
\end{align}
Furthermore, under Regime (R-i), the following holds: 
\begin{align}
&1- P_D \nonumber\\
&\sim \frac{1}{\sqrt{\pi}} \frac{\Gamma\!\left(\frac{n+1}{2}\right)}{  \Gamma\!\left(\frac{n}{2}\right)} \sqrt{\frac{2}{n+1}} \, \left(\Gamma\!\left(\frac{1}{2}, \frac{n+1}{2r_2^2}\right) \hspace{-0.003cm}- \hspace{-0.003cm}  \Gamma\!\left(\frac{1}{2}, \frac{n+1}{2r_1^2}\right)\right). \label{eq:bdR2fixed} 
\end{align}  
To show~\eqref{eq:uppbdapp}, we use the expression of $1 - P_D$ as given by equation~(\ref{eq:1-pd}) and the fact that that $\left(1 + u\right)^{-\frac{n + 1}{2}} \geq e^{-\frac{n + 1}{2} u}$ for $u \geq 0$, and we write: 
\begin{align}
 &1- P_D \nonumber\\
 &= \frac{1}{\sqrt{\pi}} \frac{\Gamma\!\left(\frac{n+1}{2}\right)}{  \Gamma\!\left(\frac{n}{2}\right)}\int_{\frac{1}{r_2^2}}^{\frac{1}{r_1^2}} \left(1 + u\right)^{-\frac{n + 1}{2}} u^{-\frac{1}{2}}\,du \nonumber\\
 &\geq \frac{1}{\sqrt{\pi}} \frac{\Gamma\!\left(\frac{n+1}{2}\right)}{  \Gamma\!\left(\frac{n}{2}\right)}\int_{\frac{1}{r_2^2}}^{\frac{1}{r_1^2}} e^{-\frac{n + 1}{2} u} u^{-\frac{1}{2}}\,du  \nonumber \\
 &= \frac{1}{\sqrt{\pi}} \frac{\Gamma\!\left(\frac{n+1}{2}\right)}{  \Gamma\!\left(\frac{n}{2}\right)} \sqrt{\frac{2}{n+1}} \, \left(\Gamma\!\left(\frac{1}{2}, \frac{n+1}{2r_2^2}\right) - \Gamma\!\left(\frac{1}{2}, \frac{n+1}{2r_1^2}\right)\right) \nonumber
\end{align}
where the last equation is due to the definition of the upper incomplete Gamma function. 
To further show~(\ref{eq:bdR2fixed}), we note that under Regime (R-i) the following holds: 
\begin{align}
&\lim_{n \rightarrow \infty} \frac{\frac{1}{\sqrt{\pi}} \frac{\Gamma\left(\frac{n+1}{2}\right)}{  \Gamma\left(\frac{n}{2}\right)}\int_{\frac{1}{r_2^2}}^{\frac{1}{r_1^2}} \left(1 + u\right)^{-\frac{n + 1}{2}} u^{-\frac{1}{2}}\,du}{\frac{1}{\sqrt{\pi}} \frac{\Gamma\left(\frac{n+1}{2}\right)}{  \Gamma\left(\frac{n}{2}\right)}\int_{\frac{1}{r_2^2}}^{\frac{1}{r_1^2}} e^{-\frac{n + 1}{2} u} u^{-\frac{1}{2}}\,du} \nonumber\\
&= \lim_{n \rightarrow \infty} \frac{\int_{\frac{1}{r_2^2}}^{\frac{1}{r_1^2}} e^{-\frac{n + 1}{2} \left[\ln(1 + u) - u\right]} e^{-\frac{n + 1}{2} u}u^{-\frac{1}{2}}\,du}{\int_{\frac{1}{r_2^2}}^{\frac{1}{r_1^2}} e^{-\frac{n + 1}{2} u} u^{-\frac{1}{2}}\,du}\nonumber\\
&\leq \lim_{n \rightarrow \infty} e^{-\frac{n + 1}{2} \left[\ln\left(1 + \frac{1}{r_1^2}\right) - \frac{1}{r_1^2}\right]} \label{eq:nonincreasing}\\
&= \left[ \lim_{n \rightarrow \infty} e^{-\frac{n + 1}{2} \ln\left(1 + \frac{1}{r_1^2}\right)} \right] \left[ \lim_{n \rightarrow \infty} e^{\frac{n + 1}{2} \frac{1}{r_1^2} }  \right]\nonumber\\
&= e^{-\frac{1}{\xi^2}} e^{\frac{1}{\xi^2}} = 1, \label{eq:r2d2}
\end{align}
where  in order to write equation~\eqref{eq:r2d2}, we used the fact that $r_1^2$ behaves according to equation~\eqref{eq:tailR2r2rmax} under Regime (R-i). Equation~\eqref{eq:nonincreasing} holds true since $\ln(1 + u) - u$ is non-increasing for $u \geq 0$.

\end{document}